\documentclass[11pt]{llncs}
\usepackage{microtype}
\usepackage{makeidx}
\usepackage{graphicx}
\usepackage{xcolor}
\usepackage{ifthen}
\usepackage{microtype}
\usepackage{boxedminipage}
\usepackage{amsmath}
\usepackage{amsfonts}
\usepackage{amssymb}
\usepackage{authblk}
\usepackage{breakcites}
\usepackage{chemscheme}
\usepackage{hyperref}
\usepackage{framed}
\usepackage{float}
\usepackage{breqn}
\usepackage{array, makecell}

\newcommand{\lv}[1]{}

\newcommand{\E}{{\bf E}}
\newcommand{\A}{{\bf A}}
\newcommand{\B}{{\bf B}}

\newcommand{\pr}{\mathbf{Pr}}
\newcommand{\eps}{{\varepsilon}}
\newcommand{\etal}{{\it et al. }}

\newcommand{\nl}{\bot}
\newcommand{\veps}{\varepsilon}
\newcommand{\Cst}{{C^\star}}

\newcommand{\ti}{{\bar i}}
\newcommand{\ostp}[1]{{O_{#1}^{'}}}

\newcommand{\R}{\mathbb{R}}

\newcommand{\Xst}{X^{\star}}

\newcommand{\cst}{c^{\star}}

\newcommand{\OPT}{OPT^{\star}}

\newcommand{\norm}[1]{\lVert #1 \rVert}
\newcommand{\ham}{\mathcal{H}}

\allowdisplaybreaks

\graphicspath{{./Figures/}}

\begin{document}

\title{Streaming {PTAS} for Binary $\ell_0$-Low Rank Approximation}
%
%
\author{Anup Bhattacharya\inst{1} \and Dishant Goyal\inst{2} \and Ragesh Jaiswal\inst{2}\thanks{This work was done while the author was on a sabbatical from IIT Delhi and visiting UC San Diego.} \and Amit Kumar\inst{2}}
%
%
%
\institute{
Indian Statistical Institute Kolkata.\thanks{Email address: \email{bhattacharya.anup@gmail.com}} \and
Department of Computer Science and Engineering, \\
Indian Institute of Technology Delhi.\thanks{Email addresses: \email{\{Dishant.Goyal, rjaiswal, amitk\}@cse.iitd.ac.in}}
}
{\def\addcontentsline#1#2#3{}\maketitle}

\begin{abstract}
We give a 3-pass, polylog-space streaming PTAS for the constrained binary $k$-means problem and a 4-pass, polylog-space streaming PTAS for the binary $\ell_0$-low rank approximation problem. 
The connection between the above two problems has recently been studied.
We design a streaming PTAS for the former and use this connection  to obtain streaming PTAS for the latter. 
This is the first constant pass, polylog-space streaming algorithm for either of the two problems. 
\end{abstract}




\pagestyle{plain}
\setcounter{page}{1}

\section{Introduction}
Low rank approximation is a common data analysis problem that has various applications in computer science.
The most general version of the problem, the $\ell_p$-low rank approximation problem, is defined in the following manner: 

\begin{quote}
\underline{\it $\ell_p$-low rank approximation}: Given a matrix $\A \in \R^{n \times d}$ (with $n \geq d$) and an integer $r$, find a rank-$r$ matrix $\B \in \R^{n \times d}$ such that $\norm{\A - \B}_p^p \equiv \sum_{i, j}|\A_{i, j} - \B_{i, j}|^{p}$ is minimised.
\end{quote}
The above definition is for any positive value of $p$. When $p =0$, the objective is to minimise $\norm{\A - \B}_0$ which is defined to be the number of mis-matches in the matrices $\A$ and $\B$. 
The $\ell_p$-low rank approximation problem is known to be $\mathsf{NP}$-hard for $p \in \{0, 1\}$ while for $p=2$ the problem can be solved using SVD (Singular Value Decomposition). 
Even though efficient approximation algorithms for these problems have been known, the approximation factor is large (polynomial in $r$). 
Recent work of Ban \etal~\cite{ban} addressed the open question about whether a PTAS is possible for these problems.
They showed that for $p \in (1, 2)$ there is no constant factor approximation algorithm running in time $2^{r^{\delta}}$ for a constant $\delta > 0$ under the small set expansion hypothesis and exponential time hypothesis (ETH). 
This shows that an exponential dependence on $r$ is necessary. 
On the upper bound side, they give a $(1+\veps)$-approximation algorithm with running time $n^{poly(\frac{r}{\veps})}$ for the cases when $p \in (0, 2)$.

The case when $p=0$ has also been studied. 
There exists $poly(r \log{n})$ bicriteria approximation algorithm for the $\ell_0$-low rank approximation problem~\cite{bkw17}. 
The problem can alternatively be stated as: given an $n\times d$ matrix $\A$, find an $n \times r$ matrix $\mathbf{U}$ and a $r \times d$ matrix $\mathbf{V}$ such that $\norm{\A - \mathbf{U} \cdot \mathbf{V}}_0$ is minimised. 
There is an interest in specific class of instances of the $\ell_0$-low rank approximation problem where the matrices $\A, \mathbf{U}, \mathbf{V}$  in the above formulation are binary matrices. 
In fact, we can generalise even further by making the notion of $\mathbf{U} \cdot \mathbf{V}$ in the above definition more flexible in the following manner: If $\A' = \mathbf{U} \cdot \mathbf{V}$, then $\A'_{ij}$ is the inner product of the $i^{th}$ row of $\mathbf{U}$ and the $j^{th}$ column of $\mathbf{V}$. We can consider various fields for this inner product. The two popularly explored fields are: (i) $\mathbb{F}_2$ with inner product defined as $\langle x, y\rangle \equiv \oplus_i (x_i \cdot y_i)$, and (ii) Boolean semiring $\{0, 1, \wedge, \vee\}$ with inner product defined as $\langle x, y\rangle \equiv \vee_i (x_i \wedge y_i) = 1 - \prod_i (1 - x_i \cdot y_i)$. There are various previous works that consider the above specific versions. 
We can generalise the problem (using the formulation in terms of $\mathbf{U}$ and $\mathbf{V}$) so that the above versions become special cases. 
This was done by Ban \etal~\cite{ban} and they called this problem  {\em generalised binary $\ell_0$-rank-$r$ problem} that is defined below.

\begin{quote}
\underline{\it Generalised binary $\ell_0$-rank-$r$}: Given a matrix $\A \in \{0, 1\}^{n \times d}$ with $n \geq d$, an integer $r$, and an inner product function $\langle.,. \rangle: \{0, 1\}^r \times \{0, 1\}^r \rightarrow \{0, 1\}$, find matrices $\mathbf{U} \in \{0, 1\}^{n \times r}$ and $\mathbf{V} \in \{0, 1\}^{r \times d}$ that minimises $\norm{\A - \mathbf{U} \cdot \mathbf{V}}_0$, where $\mathbf{U} \cdot \mathbf{V}$ is computed using the inner product function. That is $[\mathbf{U} \cdot \mathbf{V}]_{ij}$ is the inner product of the $i^{th}$ row of $\mathbf{U}$ with the $j^{th}$ column of $\mathbf{V}$.
\end{quote}
Ban \etal~\cite{ban} showed that there is no approximation algorithm for the generalised binary $\ell_0$-rank-$r$ problem running in time $2^{2^{\delta r}}$ for a constant $\delta > 0$.
The work of Ban \etal~\cite{ban} and Fomin \etal~\cite{fomin18} addressed one of the main open questions for generalised binary $\ell_0$ rank-$r$ problem -- whether a PTAS for constant $r$ is possible.
They give such a PTAS using very similar set of ideas (even though they were obtained independently).
We extend the previous work of Ban \etal and Fomin \etal to the streaming setting by using the connection of this problem to the constrained binary $k$-means problem. 
This connection was given and used by both Ban \etal~\cite{ban} and by Fomin \etal~\cite{fomin18}. 
We will now talk about the constrained binary $k$-means problem and its connection to generalised binary $\ell_0$ rank-$r$  problem.
We will start with the binary $k$-means problem which is an interesting problem on its own.
\begin{quote}
\underline{\it Binary $k$-means}: Given a set of points $X \subseteq \{0, 1\}^d$ and a positive integer $k$, find a set of $k$ centers $C \subseteq \{0, 1\}^d$ such that the $k$-means cost function $\Phi(C, X) \equiv \sum_{x \in X} \min_{c \in C} \norm{x - c}_2^2$ is minimised.
\end{quote}
Note that $C$ and $X$ are restricted to be from the set $\{0, 1\}^d$, the elements of which can alternatively be interpreted as $d$-bit strings. 
So, the squared Euclidean distance between any two points $x, y \in \{0, 1\}^d$ that is used in the $k$-means cost function can alternatively be written as $\norm{x - y}_2^2 = \ham(x, y)$, where $\ham(x, y)$ denotes the {\em Hamming distance} between strings $x$ and $y$\footnote{The hamming distance between two binary strings of equal length is the number of bits on which they differ.}. 
Note that $(\{0, 1\}^d, \ham)$ is a metric which means, in particular, that the distance function $\ham$ over the domain $\{0, 1\}^d$ satisfies the triangle inequality:
\begin{equation}\label{eqn:triangle}
\forall x, y, z \in \{0, 1\}^d, \ \ \ham(x, z) \leq \ham(x, y) + \ham(y, z).
\end{equation}
The binary $k$-means problem has been examined in the past by Kleinberg, Papadimitriou, and Raghavan~\cite{kpr04}, Ostrovsky and Rabani~\cite{or02}, and Alon and Sudakov~\cite{as99}\footnote{Alon and Sudakov~\cite{as99} consider the dual maximization problem that involves maximising $nd - \Phi(C, X)$ instead of minimizing $\Phi(C, X)$.}.
Ostrovsky and Rabani~\cite{or02} gave a $(1+\veps)$-approximation algorithm with running time $n^{f(k, \veps)}$ for some function $f$.
Fomin \etal~\cite{fomin18} gave an EPTAS (under the assumption that $k$ is a constant) with running time $f(k, \veps) \cdot n^{O(1)}$ for some function $f$. 
In fact, Fomin \etal~\cite{fomin18} gave an EPTAS for a more generalised version of the binary $k$-means problem that they call the {\em constrained binary $k$-means problem} and the EPTAS for the binary $k$-means problem trivially follows from this. We talk about this problem next.

We will work with the definition of the constrained binary $k$-means problem given by Fomin \etal~\cite{fomin18}. 
For this, we first need to define the concept of a set of $k$ centers $C \subseteq \{0, 1\}^d$ satisfying a set of $k$-ary relations. 
Given a set $\mathcal{R} = \{R_1, ..., R_d\}$ of $n$, $k$-ary binary relations (i.e., $R_i \subseteq \{0, 1\}^k$ for every $i$), a set $C = \{c_1, ..., c_k\} \subseteq \{0, 1\}^d$ of $k$ centers is said to satisfy $\mathcal{R}$ iff $(c_1[i], ..., c_k[i]) \in R_i$ for every $i = 1, ..., d$. Here, $c_j \in \{0, 1\}^d$ is thought of as a $d$-dimensional vector and $c_j[i]$ denotes the $i^{th}$ coordinate of this vector.
We can now define the constrained binary $k$-means problem.

\begin{quote}
\underline{\it Constrained binary $k$-means}:
Given a set of $n$ points $X \subseteq \{0, 1\}^d$, a positive integer $k$, and a set of $k$-ary relation s $\mathcal{R} = \{R_1, ..., R_d\}$, find a set of $k$ centers $C \subseteq \{0, 1\}^d$ satisfying $\mathcal{R}$ such that the cost function $\Phi(C, X) \equiv \sum_{x \in X} \min_{c \in C} \norm{x - c}_2^2 = \sum_{x \in X} \min_{c \in C} \ham(x, c)$ is minimised.
\end{quote}
Note that the distance between point $x$ and center $c$ is given as $\norm{x - c}_2^2$ as opposed to $\left( \sum_{i=1}^d |x[i] - c[i]| \right)$ by Fomin \etal~\cite{fomin18}. 
However, the formulations are equivalent since the distances are the same when $x, c \in \{0, 1\}^d$ are binary vectors.
It is important to distinguish between the definition of constrained binary $k$-means problem given above with the {\em constrained $k$-means problem} that has been a known problem in the $k$-means clustering literature.

\paragraph{Comparison with constrained $k$-means problem} Ding and Xu~\cite{dx15} gave a unified framework for constrained versions of the $k$-means problem. 
As per this framework, a constrained version of the $k$-means problem is defined by an input set $X$, a positive integer $k$, and a set of constraints {\em on the clusters}. 
The goal is to output a clustering of the dataset $X \subseteq \R^d$ satisfying the cluster constraints such that the $k$-means cost function with respect to centers defined by the centroid of clusters, is minimised. 
That is, find a clustering $\{X_1, ..., X_k\}$ satisfying the cluster constraints such that the following cost function gets minimised: $\Psi(X_1, ..., X_k) \equiv \sum_{i=1}^{k} \Phi(\mu(X_i), X_i)$, where $\mu(X_i) \equiv \frac{\sum_{x \in X_i} x}{|X_i|}$ is the centroid of the points in $X_i$. 
This way of defining constrained $k$-means has the advantage that a number of known constrained versions of the $k$-means/median problem fit into the framework.
For instance, consider the {\em $r$-gather clustering problem} where the constraint on the cluster is that every cluster should contain at least $r$ points. See \cite{dx15} for a more elaborate list of problems that fit into the unified framework of Ding and Xu.
This kind of generalisation raises questions about the representation and conciseness of the cluster constraints. 
This is an important consideration while defining a unified framework since the number of possible clusterings of the data can be very large. 
This problem is resolved by defining a {\em partition algorithm} as an alternative to stating all valid clusterings for a particular problem. 
A partition algorithm, when given a center set $\{c_1, ..., c_k\} \subseteq \R^d$, outputs a valid clustering $X_1, ..., X_k$ such that $\sum_{i=1}^{k} \Phi(c_i, X_i)$ is minimised\footnote{Note that for the unconstrained or standard $k$-means problem, the partition algorithm is simply the {\em Voronoi partitioning algorithm}.}. 
Interestingly, for many of the constrained versions of the $k$-means problem, such as the $r$-gather/capacity problem, there is a partition algorithm available. 
One of the contributions of Ding and Xu~\cite{dx15} was to give such partition algorithms for a number of constrained clustering problems. 
The other main contribution of Ding and Xu~\cite{dx15} was to give a polynomial time approximation scheme (PTAS) for the constrained $k$-means problem. 
This means that one gets a PTAS for {\em any} constrained version of the $k$-means problem as long as there is an efficient partition algorithm for that version of the $k$-means problem. Subsequently, Bhattacharya \etal~\cite{bjk} gave a faster PTAS while Goyal \etal~\cite{gjk}  gave a streaming PTAS  for the problem, both using the $D^2$-sampling technique. 
At a high-level, the constrained binary $k$-means problem that we defined earlier, seems to be yet another constrained version of the $k$-means problem. 
So, the relevant question in the context of the current discussion is: 
\begin{quote}
{\it Does the constrained binary $k$-means problem fit into the unified framework of Ding and Xu~\cite{dx15}?}
\end{quote}
If the answer to the above question were yes, then a streaming PTAS for the constrained binary $k$-means problem would trivially follow from the recent work of Goyal \etal~\cite{gjk}. 
Unfortunately, this is not true. 
Note that the framework of Ding and Xu~\cite{dx15}
defines the constraints on the clusters while the definition of constrained binary $k$-means problem defines constraints on the centers. 
However, we note that the $D^2$-sampling based techniques of \cite{bjk,gjk} can be extended to this setting.

Goyal \etal~\cite{gjk} used a constant factor approximation algorithm for the standard $k$-means problem as a subroutine to obtain a PTAS for the constrained $k$-means problem within the framework of Ding and Xu~\cite{dx15}. 
Since their algorithm is based on a simple sampling idea, their algorithm can be converted to a constant pass streaming algorithm using the {\em reservoir sampling} technique.
In this work, we try to extend the same ideas to the constrained binary $k$-means problem.
However, to make the sampling ideas work, we will need additional results. 
First, we will need a streaming algorithm that gives a constant factor approximate solution for the binary $k$-means problem (i.e., the unconstrained problem). 
Second, we will need a result that says that it is possible to obtain good constrained centers of any target clusters if we have uniform samples from each of the clusters. 
Fortunately, both the results are already known. We discuss these next.
Let us start with the streaming constant factor approximation algorithms for the binary $k$-means problem.

\paragraph{Streaming constant approximation for binary $k$-means}
The binary $k$-means problem is basically the unconstrained version of the constrained binary $k$-means problem.
The following result follows from the work of Braverman \etal~\cite{brav11} on the standard discrete version of the $k$-means problem over arbitrary metric spaces.

\begin{theorem}\label{thm:bin-k-means}
There is a constant factor streaming algorithm for the binary $k$-means problem that runs in one pass over the data while storing $O(k \log{n})$ points in memory with overall running time $O(nk \log{n})$.
\end{theorem}
Let us now discuss the second issue about obtaining good centers using uniform samples from every cluster.

\paragraph{Good centers from uniform samples} We will now discuss the possibility of obtaining good centers for the constrained binary $k$-means problem using uniform samples from each of the target clusters. 
Note that for the standard $k$-means problem, this is possible since the centroid of a uniformly sampled set of $O(\frac{1}{\veps})$ points from any cluster gives a good center for that cluster. 
A good center here means that the cost with respect to this center is within $(1+\veps)$-factor of the cost with respect to the optimal center of the cluster. This result follows from a result of Inaba \etal~\cite{inaba}. 
In the context of the constrained binary $k$-means problem, just having a uniform sample from each of the clusters may not be sufficient to obtain good centers for the clusters in the $(1+\veps)$-approximation sense. 
However, Fomin \etal~\cite{fomin18} and Ban \etal~\cite{ban} showed that if one has an estimate of the size of the optimal clusters in addition to certain minimum number of uniform samples from them, then one can obtain good constrained centers for the clusters in the $(1+\veps)$-approximation sense.
The following theorem states the result of Fomin \etal~\cite{fomin18} formally. 
A similar sampling result was given and used by Ban \etal~\cite{ban}. 
We use the formulation of Fomin \etal

\begin{theorem}[Follows from Fomin \etal~\cite{fomin18}]\label{thm:fomin}
For a given instance $J = (X, k, \mathcal{R})$ of the problem, let $X_1, ..., X_k$ denote an arbitrary partition of the points in $X$, and let $\{c_1, ..., c_k\}$ be a center set satisfying $\mathcal{R}$ that minimises $\sum_{i=1}^{k} \sum_{x \in X_i} \ham(x, c_i)$.
Let $\veps > 0$ and $\tau = \Theta(\frac{k}{\veps^2} \log{\frac{1}{\veps}})$. Let $w_1, ..., w_k \in \mathbb{N}$ be such that for every $i$, $|X_i| \leq w_i \leq (1 + \frac{\veps}{2}) \cdot |X_i|$.
Let $S_1, ..., S_k$ denote a multiset of the points such that for every $i$, $S_i$ consists of $\tau$ points from $X_i$ sampled independently and uniformly with replacement. 
Then there is a simple algorithm ${\tt best}_{\mathcal{R}}(J, S_1, ..., S_k, w_1, ..., w_k)$ given below that outputs a center $C$ satisfying $\mathcal{R}$ such that $\E[\Phi(C, X)] \leq (1+\veps) \cdot \left(\sum_{i=1}^{k} \sum_{x \in X_i} \ham(x, c_i) \right)$.
\begin{framed}
${\tt best}_{\mathcal{R}}(J, S_1, ..., S_k, w_1, ..., w_k)$\\
\hspace*{0.3in} - For every $i \in \{1, ..., d\}$:\\
\hspace*{0.6in} - Let $(b_1, ..., b_k) \leftarrow \arg\min_{(b_1, ..., b_k) \in R_i} {f_i(b_1, ..., b_k)}$, where \\
\hspace*{0.8in} $f_i(b_1, ..., b_k) \equiv \sum_{j: b_j=1} w_j \cdot |S_j^{(i, 0)}| + \sum_{j: b_j = 0} w_j \cdot |S_j^{(i, 1)}|$, and \\
\hspace*{0.8in} $S_j^{(i, b)}$ denotes the set of strings in $S_j$ that has the $i^{th}$ bit as $b$.\\
\hspace*{0.6in} - $(c_1[i], c_2[i], ..., c_k[i]) \leftarrow (b_1, ..., b_k)$\\
\hspace*{0.3in} - return $(c_1, ..., c_k)$
\end{framed}
The running time of the algorithm is $O(kd \tau  + kd \cdot \sum_i |R_i|)$.
\end{theorem}
Note that the statement of the above lemma deviates from the statement of its {\it parent} lemma by Fomin \etal~\cite{fomin18}. 
The lemma by Fomin \etal~\cite{fomin18} is not for an arbitrary partition $X_1, ..., X_k$ of the given dataset $X$ and center set $\{c_1, ..., c_k\}$ satisfying $\mathcal{R}$ that optimises $\sum_{i=1}^{k} \sum_{x \in X_i} \ham(x, c_i)$. It is for the optimal partitioning $\Xst_1, ..., \Xst_k$. Note that for the optimal partitioning $\Xst_1, ..., \Xst_k$ and its corresponding optimal center set $\{\cst_1, ..., \cst_k\}$, we have $\sum_{i=1}^{k} \sum_{x \in \Xst_i} \ham(x, \cst_i) = OPT$. 
So, the final bound on the expectation given by Fomin \etal~\cite{fomin18} is in terms of $OPT$.
As far as the proof of the above theorem is concerned, we comment that the proof is essentially the same since the proof of Fomin \etal does not use the optimality of the partition $\Xst_1, ..., \Xst_k$ and in fact the proof holds for any partition $X_1, ..., X_k$ as stated in the theorem above.

The above theorem tells us that as long as we can guess the size of the target clusters and obtain $\tau$ uniform samples from every cluster, we should be able to obtain good centers for these clusters.
For the size of the clusters, since we only need an estimate within a multiplicative factor of $(1+\frac{\veps}{2})$, we can employ a brute-force strategy of trying out all the $\left(\log_{(1+\frac{\veps}{2})}{n}\right)^k$ possibilities. This brute-force strategy contributes a multiplicative factor of $(\log{n})^k$ in the running time of the PTAS\footnote{
Note that Fomin \etal~\cite{fomin18} have managed to remove this factor from the running time of their PTAS using a {\it peeling strategy} that handle optimal clusters in a sequence of iterations. However, using this peeling strategy to design a streaming PTAS that works in a few passes becomes difficult.
}.
As for obtaining $\tau$ uniform samples from each of the target clusters, this is a more tricky issue and it is not  immediately clear how to do this. 
If all the clusters are roughly of equal size, then one can uniformly sample $poly(\frac{k}{\veps})$ points from $X$ with replacement and then try out all possible $\tau$-sized (multi) subsets. 
If all clusters are roughly of equal size, we can argue that one of the (multi) subsets $S_1', ..., S_k'$ so attempted, will be $\tau$-sized uniform samples from $X_1, ..., X_k$. 
However, if the optimal clusters are of very different sizes, then uniform sampling from $X$ will clearly not work since the chance of sampling from a very small sized cluster may be very small. 
$D^2$-sampling has turned out to be a very useful tool in such cases where uniform sampling does not work. 
Given a set of centers $C$, the idea is to sample points with probability proportional to the squared distance of the points from the closest center in $C$. This boosts the probability of sampling from small sized clusters that does not have a representative in the center set $C$. 
This idea suggests an iterative way of picking good centers in $k$ rounds where one argues that either there is a good chance of picking good centers from uncovered clusters in subsequent rounds or that the current set of centers gives good $k$-means cost. 
This idea, however, is not likely to lead to a streaming algorithm with few passes, especially in the current context where uniform samples from all clusters are needed simultaneously to meet the constraints.
So, we use the idea of Bhattacharya \etal~\cite{bjk} that was subsequently used by Goyal \etal~\cite{gjk} to obtain streaming algorithms for various constrained versions of the $k$-means problems.
The main idea is to start with a constant factor approximate solution $B$ for the binary $k$-means problem and then consider the data points $X_1', ..., X_k'$ constructed (though not explicitly algorithmically) using points from $X_1, ..., X_k$ and $B$. 
The point sets $X_1', ...., X_k'$ has two advantages over the optimal partition $X_1, ..., X_k$ -- (i) good centers for $X_1', ..., X_k'$ are also good for $X_1, ..., X_k$, and (ii) it is possible to simultaneously obtain $\tau$ uniform samples from each of $X_1', ..., X_k'$.
Along with the result of Fomin \etal~\cite{fomin18} (Theorem~\ref{thm:fomin}), this should be sufficient to obtain good centers for the constrained binary $k$-means problem.
The PTAS based on the above ideas is extremely simple and can be stated as a short pseudocode given below.

\begin{framed}
{\tt GoodCenters}($J, B, \veps$)\\
\hspace*{0.6in} {\bf Inputs}: Input instance $J=(X, k, \mathcal{R})$, $\alpha$-approximate $B$, and accuracy $\veps$\\
\hspace*{0.6in} {\bf Output}: A center set $D$ satisfying $\mathcal{R}$\\
\hspace*{0.6in} {\bf Constants}: $\eta = \Theta(\frac{\tau \alpha k}{\veps^2}); \tau = \Theta(\frac{k}{\veps^2} \log{\frac{k}{\veps}})$; $\zeta = \Theta(\frac{1}{\veps})$\\
\hspace*{0.2in} (1) \ \ \ $cost \leftarrow nd$; $\mathcal{C} \leftarrow \emptyset$\\
\hspace*{0.2in} (2) \ \ \ Repeat $\zeta$ times:\\
\hspace*{0.2in} (3)\hspace*{0.2in}  \ \ \ Sample a multi-set $M$ of $\eta k$ points from $X$ using $D^2$-sampling w.r.t. center set $B$\\
\hspace*{0.2in} (4)\hspace*{0.2in}  \ \ \ $M \leftarrow M \cup$ \{$\tau k$ copies of each element in $B$\}\\
\hspace*{0.2in} (5)\hspace*{0.2in} \ \ \ For all disjoint subsets $S_1, ..., S_k$ of $M$ such that $\forall i, |S_i| = \tau$:\\
\hspace*{0.2in} (6)\hspace*{0.5in} \ \ \ For all $w_1, ..., w_k \in \{1, (1+\frac{\veps}{2})^1, ..., (1+\frac{\veps}{2})^{\lceil \log_{1+\frac{\veps}{2}}{n} \rceil}\}$:\\
\hspace*{0.2in} (7)\hspace*{0.9in} $\mathcal{C} \leftarrow \mathcal{C} \cup {\tt best}_{\mathcal{R}}(J, S_1, ...,S_k, w_1, ..., w_k)$\\
\hspace*{0.2in} (8) \ \ \ return the $k$-center set from $\mathcal{C}$ with least cost
\end{framed}

\noindent
We prove the following result with respect to the algorithm given above: 

\begin{theorem}\label{thm:main-theorem}
Let $0 < \veps \leq \frac{1}{2}$. For any constrained binary $k$-means instance $J = (X, k, \mathcal{R})$, let $B$ denote an $\alpha$-approximate solution for the binary $k$-means problem instance $(X, k)$. The algorithm ${\tt GoodCenters}(J, B, \veps)$, returns a center set $D$ satisfying $\mathcal{R}$ such that: 
\[
\pr[\Phi(D, X) \leq (1 + \veps) \cdot OPT] \geq \frac{3}{4}.
\]
The running time of the algorithm is $O\left( nd \cdot (\log{n})^k \cdot 2^{\tilde{O}(\frac{k^2}{\veps^2})}\right)$.
\end{theorem}

We prove the above theorem in Section~\ref{sec:ptas-analysis}. 
Let us now see how this algorithm gives a 3-pass streaming PTAS for the problem.
The first pass of the algorithm is used to find a constant factor approximate solution $B$ to the binary $k$-means problem $(X, k)$ corresponding to the given instance $J = (X, k, \mathcal{R})$. 
From the discussion earlier, we know that there is a one pass algorithm that returns a constant factor approximate solution to the binary $k$-means problem and that uses $O(k \log{n})$ space. 
In the second pass, we execute lines (1-7) of the algorithm ${\tt GoodCenters}(J, B, \veps)$. 
This can indeed be executed in a single pass. 
Note that line (2) is for probability amplification and all $\zeta$ iterations can be executed independently. 
The main step is line (3) where we need to $D^2$-sample $\eta k$ points w.r.t. center-set $B$ independently from $X$.
This can be done using {\em reservoir sampling} in a single pass\footnote{{\it Reservoir sampling}: Let $y_i$ denote the squared distance of points $x_i$ to the nearest center in the center set $B$. Then one $D^2$-sample can be obtained while making a pass over the data in the following manner: Store the first point and on seeing the $i^{th}$ point for $i > 1$, replace the stored element with probability $\frac{y_i}{\sum_{j=1}^{i} y_j}$ and continue with the remaining probability.}.
Lines (5-7) accumulates $k$-center sets in $\mathcal{C}$ corresponding to all possible subsets and all possible choices for $w_1, ..., w_k$. 
The final step of line (8) involves picking the $k$-center set from $\mathcal{C}$ with the least cost. 
We need one more pass over the data to perform line (8). 
This makes a total of 3 passes.
The space requirement for the $3^{rd}$ pass is $O \left(d \cdot (\log{n})^k \cdot 2^{\tilde{O}(\frac{k^2}{\veps^2})} \right)$. We summarise our result formally as the theorem below the proof of which follows trivially from the discussion above.

\begin{theorem}[Main result for constrained binary $k$-means]\label{thm:bin-main}
Let $0 < \veps \leq 1/2$. There is a 3-pass streaming algorithm that outputs a $(1+\veps)$-approximate solution for any instance of the constrained binary $k$-means problem. The space and per-item processing time of our algorithm is $O \left(d \cdot (\log{n})^k \cdot 2^{\tilde{O}(\frac{k^2}{\veps^2})} \right)$.
\end{theorem}

Note that as per the formulation of the constrained binary $k$-means problem, the output is supposed to be a set of $k$ centers. 
The above 3-pass algorithm outputs such a $k$ center set $D$. 
However, if the objective is to output the clustering of the data points $X$ with respect to $D$, then one more pass over the data will be required and the resulting algorithm will be a 4-pass algorithm. 
This is relevant for the $\ell_0$-rank-$r$ approximation problem that we discuss next.
We obtain a result for the generalised binary $\ell_0$-rank-$r$ problem that is similar to the above result, using a simple reduction to the constrained binary $k$-means problem. 
This reduction is used by both Fomin \etal~\cite{fomin18} and Ban \etal~\cite{ban}. We restate the result of Fomin \etal~\cite{fomin18} for clarity.

\begin{lemma}[Lemma 1 and 2 of \cite{fomin18}]
For any instance $(\A, r)$ of the generalised binary $\ell_0$-rank-$r$ approximation problem, one can construct in time $O(n + d + 2^{2r})$ an instance $(X, k = 2^r, \mathcal{R})$ of constrained binary $k$-means problem with the following property: Given any $\alpha$-approximate solution $C$ of $(X, k, \mathcal{R})$, an $\alpha$-approximate solution $\B$ of $(\A, r)$ can be constructed in time $O(rnd)$.
\end{lemma}
The dataset $X$ corresponding to matrix $\A$, in the above reduction, is essentially the rows of the matrix $\A$ and $\forall i, R_i = \{(\langle x, \lambda_1 \rangle, ..., \langle x, \lambda_k \rangle) : x \in \{0, 1\}^r\}$ and $\lambda_i$'s are pairwise distinct vectors in $\{0, 1\}^r$. 
The above reduction and Theorem~\ref{thm:bin-main} gives the following main result for the generalised binary $\ell_0$-rank-$r$ approximation problem. 
Note that since we need to output a matrix $\B$, we will need the clustering of the rows of $\A$ and as per previous discussion this will require one more pass than that in Theorem~\ref{thm:bin-main}.

\begin{theorem}[Main result for generalised binary $\ell_0$-rank-$r$ approximation]
Let $0 < \veps \leq 1/2$. There is a 4-pass streaming algorithm that makes row-wise passes over the input matrix and outputs a $(1+\veps)$-approximate solution for any instance of the generalised binary $\ell_0$-rank-$r$ problem. 
The space and per-item processing time of our algorithm is $O \left(d \cdot (\log{n})^{2^r} \cdot 2^{\tilde{O}(\frac{2^{2r}}{\veps^2})} \right)$.
\end{theorem}

A lot of work has been done for the $k$-means problem and the $\ell_p$ low rank approximation problems. 
The following related work subsection will help see our work in the right perspective.

\subsection{Related work}
The binary $k$-means problem is a special case of the discrete variant of the classical $k$-means problem where the clustering problem is defined over the metric $(\{0, 1\}^n, \mathcal{H})$. 
The problem was introduced and studied by Kleinberg, Papadimitriou, and Raghavan~\cite{kpr98,kpr04} by the name of {\em segmentation problems}. 
They showed that the problem is $\mathsf{NP}$-hard for $k > 1$ and gave approximation algorithms for the dual maximisation problem where the goal is to maximise $nd - \Phi(C, X)$.
Ostrovsky and Rabani~\cite{or02} gave a randomised PTAS with running time $n^{f(k, \veps)}$ for some function $f$. 
More recently, Fomin \etal~\cite{fomin18} gave an efficient PTAS with running time $g(k, \veps) \cdot n^{O(1)}$ for some function $g$. 
In fact, Fomin~\etal gave such an efficient PTAS for a much more generalised version called the generalised constrained binary $k$-means problem which we discussed earlier. 
Ban~\etal~\cite{ban} independently obtained similar results\footnote{The result of Ban~\etal~\cite{ban} does not explicitly discuss the binary $k$-means problem.}.

The binary versions of the $\ell_0$-low rank approximation problem is relevant in a number of contexts (e.g., \cite{gutch10,painsky,dan,bv10,sbm03,singliar}).
Even though the generalised $\ell_0$-rank-$r$ approximation problem was named more recently by Ban~\etal~\cite{ban}, the special cases of the problem have been studied in the past.
For instance, it known that for the special case where the field for the inner product is $GF(2)$, the problem is $\mathsf{NP}$-hard for every $r \geq 1$~\cite{gv15,dan}.
Various constant factor approximation algorithms have been given for cases where $r$ is a fixed constant~\cite{sjy09,jphy14,bkw17}.
There also exist $O(r)$-approximation algorithm in time $n^{O(r)}$~\cite{dan}.
Ban~\etal~\cite{ban} showed a hardness-of-approximation result conditioned on the Exponential Time Hypothesis (ETH) showing that there is no approximation algorithm (beyond a fixed constant) for the generalised binary $\ell_0$-rank-$r$ approximation problem running in time $2^{2^{\delta r}}$ for a constant $\delta >0$.
They support their lower bound with a PTAS that runs in time $(\frac{2}{\veps})^{\frac{2^{O(r)}}{\veps^2}} nd^{1 + o(1)}$. A similar PTAS (using similar ideas) was given independently by Fomin~\etal~\cite{fomin18}.


\section{Analysis of PTAS (Proof of Theorem~\ref{thm:main-theorem})}\label{sec:ptas-analysis}
In this section, we prove our main result related to the algorithm {\tt GoodCenters}. 
We state the pseudocode for {\tt GoodCenters} and the statement of the theorem for ease of reading.

\begin{framed}
{\tt GoodCenters}($J, B, \veps$)\\
\hspace*{0.6in} {\bf Inputs}: Input instance $J=(X, k, \mathcal{R})$, $\alpha$-approximate $B$, and accuracy $\veps$\\
\hspace*{0.6in} {\bf Output}: A center set $D$ satisfying $\mathcal{R}$\\
\hspace*{0.6in} {\bf Constants}: $\eta = \Theta(\frac{\tau \alpha k}{\veps^2}); \tau = \Theta(\frac{k}{\veps^2} \log{\frac{k}{\veps}})$; $\zeta = \Theta(\frac{1}{\veps})$\\
\hspace*{0.2in} (1) \ \ \ $cost \leftarrow nd$; $\mathcal{C} \leftarrow \emptyset$\\
\hspace*{0.2in} (2) \ \ \ Repeat $\zeta$ times:\\
\hspace*{0.2in} (3)\hspace*{0.2in}  \ \ \ Sample a multi-set $M$ of $\eta k$ points from $X$ using $D^2$-sampling w.r.t. center set $B$\\
\hspace*{0.2in} (4)\hspace*{0.2in}  \ \ \ $M \leftarrow M \cup$ \{$\tau k$ copies of each element in $B$\}\\
\hspace*{0.2in} (5)\hspace*{0.2in} \ \ \ For all disjoint subsets $S_1, ..., S_k$ of $M$ such that $\forall i, |S_i| = \tau$:\\
\hspace*{0.2in} (6)\hspace*{0.5in} \ \ \ For all $w_1, ..., w_k \in \{1, (1+\frac{\veps}{2})^1, ..., (1+\frac{\veps}{2})^{\lceil \log_{1+\frac{\veps}{2}}{n} \rceil}\}$:\\
\hspace*{0.2in} (7)\hspace*{0.9in} $\mathcal{C} \leftarrow \mathcal{C} \cup {\tt best}_{\mathcal{R}}(J, S_1, ...,S_k, w_1, ..., w_k)$\\
\hspace*{0.2in} (8) \ \ \ return the $k$-center set from $\mathcal{C}$ with least cost
\end{framed}

\noindent
Following is the restatement of the main result with respect to the algorithm above.

\begin{theorem}
Let $0 < \veps \leq \frac{1}{2}$. For any constrained binary $k$-means instance $J = (X, k, \mathcal{R})$, let $B$ denote an $\alpha$-approximate solution for the binary $k$-means problem $(X, k)$. The algorithm ${\tt GoodCenters}(J, B, \veps)$, returns a center set $D$ satisfying $\mathcal{R}$ such that: 
\[
\pr[\Phi(D, X) \leq (1 + \veps) \cdot OPT] \geq \frac{3}{4}.
\]
The running time of the algorithm is $O\left( nd \cdot (\log{n})^k \cdot 2^{\tilde{O}(\frac{k^2}{\veps^2})}\right)$.
\end{theorem}
We will need a few definitions for our analysis with respect to the given instance $J = (X, k, \mathcal{R})$.
Let $C = \{c_1, ..., c_k\}$ denote the optimal center set satisfying $\mathcal{R}$ and let $X_1, ..., X_k$ denote the corresponding clustering induced by $C$. 
That is, for every $i, X_i = \{x \in X : \arg\min_{j} \ham(c_j, x) = i\}$. 
For every $i$, let $\Delta_i \equiv \Phi(c_i, X_i) = \sum_{x \in X_i} \ham(x, c_i)$ and as before let $OPT = \sum_i \Delta_i$.
We will also need terms related to the optimal solution to the corresponding binary $k$-means instance $(X, k)$ (that is, the corresponding unconstrained instance). Let $\Cst$ denote the optimal binary $k$-means solution, let $\Xst_1, ..., \Xst_k$ denote the corresponding clustering induced by $\Cst$ and let $\OPT = \Phi(\Cst, X)$. We know the following about some of the quantities defined above and the center set $B$ given as input to the algorithm:
\begin{equation}\label{eqn:basicP}
\OPT \leq OPT \qquad \textrm{and} \qquad \Phi(B, X) \leq \alpha \cdot \OPT. 
\end{equation}
The first inequality follows from the fact that $OPT$ denotes the optimal solution to the constrained version as opposed to $\OPT$ that denotes the optimal solution to the unconstrained version. The second inequality follows from the fact that $B$ is an $\alpha$-approximate solution to the binary $k$-means instance $(X, k)$. 
The outer iteration (repeat $\zeta$ times in line (2)) is for probability amplification. 
We will show that the probability of finding a good center set in one iteration is $\Omega(\veps)$ and the theorem will follow from simple probability calculations. 
Let us now focus on a single iteration of the algorithm. 
We will assume that we know $w_1, ..., w_k$ such that $\forall i, |X_i| \leq w_i \leq (1+\frac{\veps}{2}) \cdot |X_i|$. Note that we will try all possibilities such that the inequalities hold. It will be easier to analyse assuming that we know the correct values. 
We will show that with probability at least $\Omega(\veps)$, there are disjoint (multi) subsets $S_1, ...., S_k$ of $M$ (see line (5))) such that,
\begin{equation}\label{eqn:desirable}
\Phi({\tt best}_{\mathcal{R}}(J, S_1, ..., S_k, w_1, ..., w_k), X) \leq (1 + \veps) \cdot OPT.
\end{equation}
Since we try out all possible subsets in line (5), we will obtain the desired result.
We argue in the following manner: consider the multi-set $B' = \{\tau k \textrm{ copies of each element in } B\}$. 
We can interpret $B'$ as a union of of multi-sets $B_1', ..., B_k'$, where $B_i' = \{\tau \textrm{ copies of each element in $B$}\}$. 
Also, since $M$ consists of $\eta k$ independently sampled points, we can interpret $M$ as a union of multi-sets $M_1', ..., M_k'$ where $M_i'$ is the $i^{th}$ bunch of $\eta$ points sampled.
For all $i$, let $M_i \equiv B_i' \cup M_i'|_{X_i}$ where $M_i'|_{X_j}$ denotes the set of points in the multi-set $M_i'$ (with repetition) that belongs to $X_i$.
We will show that there are subsets $S_i \subseteq M_i$ for every $i$ such that eqn. (\ref{eqn:desirable}) holds.
We state this formally as the next lemma which we will prove in the remaining section. 

\begin{lemma}\label{lemma:toshow}
Let multi-sets $M_1, ..., M_k$ be as defined above. Then
\[
\pr \left[ \exists S_1, ..., S_k \textrm{ s.t. } \forall i, S_i \subseteq M_i \textrm{ and } \forall i, |S_i| = \tau \textrm{ and } S_1,..., S_k \textrm{ satisfies eqn. (\ref{eqn:desirable})}\right] = \Omega(\veps).
\]
\end{lemma}
We divide the cluster indices into two groups based on the value of $\Phi(B, X_i)$ and then do a case analysis. Let 
\[
I_s = \left\{ j : \Phi(B, X_j) \leq \frac{\veps}{6 \alpha k} \cdot \Phi(B, X) \right\} \quad \textrm{and} \quad I_l = \left\{ j : \Phi(B, X_j) > \frac{\veps}{6 \alpha k} \cdot \Phi(B, X) \right\}
\]
We will show that the good $S_i$'s as in Lemma~\ref{lemma:toshow} are such that for $j \in I_s, S_j \subseteq B_j' \subseteq M_j$ and for $j \in I_l, S_j \subseteq M_j$.
The analysis share similarity with the analysis of the $D^2$-sampling based algorithm in the context of the standard $k$-means problem of Bhattacharya \etal~\cite{bjk} and Goyal \etal~\cite{gjk}.
However, there are significant deviations and the arguments have to be adapted to the current setting. 
The main difference is because of the fact that in the context of the standard $k$-means problem, a uniform sample from a cluster was sufficient to obtain a good center from that cluster.
So, one could argue approximation guarantee {\em cluster-wise}.
In the current context, one needs to argue simultaneously with respect to all clusters. 
Even though some parts of the proof may be similar to the previous works~\cite{bjk,gjk}, there are significant differences and we give the detailed proofs here.

Let us consider the index set $I_s$ first. 
Let $j$ be any index in the set $I_s$. 
We will show that there is a set $X_j'$ consisting only of elements in the set $B$ such that a uniform sample from $X_j'$ (along with similar samples from other $X_i$'s) will give a good center-set.
For any point $x \in X$, let $n(x)$ denote the center in the set $B$ that is closest to $x$. That is, $n(x) \equiv \arg\min_{b \in B} \ham(x, b)$. 
Let us define the multi-set $X_j'$ as 
\begin{equation}\label{eqn:X_j'1}
X_j' = \{n(x) : x \in X_j\} \quad \textrm{ for every $j \in I_s$}
\end{equation}
That is, we take the nearest centers of elements in $X_j$ with appropriate multiplicities to construct $X_j'$. 
The intuition behind constructing the set $X_j'$ is that since the cost of $X_j$ with respect to center-set $B$ is very small, the points in the set $X_j$ are close to the centers in $B$ and hence the centers in $B$ can act as ``proxy" for the points in the set $X_j$. Obtaining a uniform sample from $X_j'$ is much simpler since we consider the set $B_j'$ that has appropriate number of copies from the set $B$.

Now constructing a similar set for an index $j \in I_l$, is a bit more involved. 
Since the cost of $X_j$ for any $j \in I_l$ is not small (as opposed to indices in $I_s$), the points from the set $B$ alone cannot act as proxy for the points in the set $X_j$. On the other hand, if we sample using $D^2$-sampling w.r.t. set $B$, then all the points in $X_j$ that are far from centers in $B$ will have a good chance of being sampled. However, the same is not true for points in $X_j$ that are close to centers in $B$. 
So, what we need to do is to consider a partition of the points in $X_j$ into {\em near} points and {\em far-away} points. 
The far-away points have a good chance of being sampled in line (3) and the centers in $B$ can act as proxy for near points. 
We define the set $X_j'$ for $j \in I_l$ more formally now. 
The closeness of point in $X_j$ to points in $B$ is quantified using radius $R_j$ that is defined by the equation:
\[
R_j \equiv \frac{\veps}{9} \cdot \frac{\Phi(B, X_j)}{|X_j|}
\]
Let $X_j^{near}$ be points in $X_j$ that are within distance $R_j$ from a center in $B$ and $X_j^{far}$ denote the remaining points. 
That is, 
\[
X_j^{near} \equiv \{x \in X_j : \min_{b\in B} \ham(x, b) \leq R_j\} \qquad \textrm{and} \qquad X_{j}^{far} \equiv X_j \setminus X_{j}^{near}.
\]
Using these, we define the multi-set $X_j'$ as:
\begin{equation}\label{eqn:X_j'2}
X_j' \equiv X_j^{far} \cup \{n(x) : x \in X_j^{near}\} \quad \textrm{ for every $j \in I_l$}.
\end{equation}
Note that $|X_j| = |X_j'|$. Let $n_j = |X_j|$ and $\bar{n}_j = |X_j^{near}|$.
Having defined the sets $X_1', ..., X_k'$ corresponding to $X_1, ..., X_k$ in eqn. (\ref{eqn:X_j'1}) and (\ref{eqn:X_j'2}), we will now try to show that a good center set for $X_1', ..., X_k'$ will also be good for $X_1, ..., X_k$. 
Let $C' = \{c_1', ..., c_k'\}$ be a center set such that $C'$ satisfies $\mathcal{R}$ and $C'$ minimises the cost $\sum_{i=1}^{k} \Phi(c_i', X_i')$.
The next few lemmas will be useful in the analysis.

\begin{lemma}\label{lemma:inter2}
For any $j \in I_l$, $\Delta_j \geq \frac{4 \bar{n}_j}{\veps} \cdot R_j$.
\end{lemma}
\begin{proof}
Let $b = \arg\min_{b \in B}{\ham(c_j, b)}$. We do a case analysis:
\begin{enumerate}
\item \underline{Case 1}: $\ham(c_j, b) \geq \frac{5}{\veps} \cdot R_j$\\
In this case, consider any point $p \in X_j^{near}$. From the triangle inequality, we have 
\[
\ham(p, c_j) \geq \ham(n(p), c_j) - \ham(n(p), p) \geq \frac{5}{\veps} \cdot R_j - R_j \geq \frac{4}{\veps} \cdot R_j.
\]
This gives $\Delta_j \geq \sum_{p \in X_j^{near}} \ham(p, c_j) \geq \frac{4 \bar{n}_j}{\veps} \cdot R_j$.

\item \underline{Case 2}: $\ham(c_j, b) < \frac{5}{\veps} \cdot R_j$\\
In this case, we have from triangle inequality:
\[
\Delta_j \geq \Phi(b, X_j) - n_j \cdot \ham(b, c_j) \geq \Phi(B, X_j) - n_j \cdot \ham(b, c_j) \geq \frac{9 n_j}{\veps} \cdot R_j - \frac{5n_j}{\veps} \cdot R_j \geq \frac{4 \bar{n}_j}{\veps} \cdot R_j.
\]
This completes the proof of the lemma.\qed
\end{enumerate}
\end{proof}

\begin{lemma}\label{lemma:combined}
$\sum_{j=1}^k \Phi(c_j', X_j') \leq OPT + \sum_{j \in I_s} \Phi(B, X_j) + \sum_{j \in I_l} \bar{n}_j R_j$.
\end{lemma}
\begin{proof}
The proof follows from the following inequalities:
\begin{eqnarray*}
\sum_j \Phi(c_j', X_j') &\leq& \sum_j \Phi(c_j, X_j') \qquad \textrm{(since $c_1', ..., c_k'$ are optimal for $X_1', ..., X_k'$)}\\
&=& \sum_{j \in I_s} \Phi(c_j, X_j') + \sum_{j \in I_l} \Phi(c_j, X_j') \\
&=& \sum_{j \in I_s} \sum_{x \in X_j} \ham(n(x), c_j) + \sum_{j \in I_l} \Phi(c_j, X_j') \qquad \textrm{(using defn. of $X_j'$)}\\
&\leq& \sum_{j \in I_s} \sum_{x \in X_j} (\ham(n(x), x) + \ham(x, c_j)) + \sum_{j \in I_l} \Phi(c_j, X_j') \qquad \textrm{(using triangle inequality)}\\
&\leq& \sum_{j \in I_s} (\Phi(B, X_j) + \Delta_j) + \sum_{j \in I_l} \Phi(c_j, X_j') \\
&=& \sum_{j \in I_s} (\Phi(B, X_j) + \Delta_j) + \sum_{j \in I_l} \left( \sum_{x \in X_j^{near}} \ham(n(x), c_j) + \sum_{x \in X_j^{far}} \ham(x, c_j) \right) \\
&& \qquad \textrm{(using defn. of $X_j^{near}$ and $X_j^{far}$)}\\
&\leq&  \sum_{j \in I_s} (\Phi(B, X_j) + \Delta_j) + \sum_{j \in I_l} \left( \sum_{x \in X_j^{near}} (\ham(n(x), x) + \ham(x, c_j)) + \sum_{x \in X_j^{far}} \ham(x, c_j)\right)\\
&& \qquad \textrm{(using triangle inequality)}\\
&=& \sum_{j \in I_s} (\Phi(B, X_j) + \Delta_j) + \sum_{j \in I_l} \left( \bar{n}_j R_j + \sum_{x \in X_j} \ham(x, c_j) \right)\\
&=& \sum_{j \in I_s} (\Phi(B, X_j) + \Delta_j) + \sum_{j \in I_l} \left( \bar{n}_j R_j + \Delta_j \right)\\
&=& OPT + \sum_{j \in I_s} \Phi(B, X_j) + \sum_{j \in I_l} \bar{n}_j R_j.
\end{eqnarray*}
This completes the proof of the lemma.\qed
\end{proof}

Let $C'' = \{c_1'', ..., c_k''\}$ be a center set that satisfies $\mathcal{R}$ such that 
\begin{equation}\label{eqn:inter1}
\sum_{i=1}^{k} \Phi(c_i'', X_i') \leq \left(1+\frac{\veps}{16} \right) \cdot \sum_{i=1}^{k} \Phi(c_i', X_i').
\end{equation}
We will now show that $C''$ is a good center set for $X$.

\begin{lemma}\label{lemma:main-lem}
$\Phi(C'', X) \leq (1 + \veps) \cdot OPT$.
\end{lemma}
\begin{proof}
The proof follows from the following sequence of inequalities: 
\begin{eqnarray*}
\Phi(C'', X) &\leq& \sum_{j=1}^{k} \Phi(c_j'', X_j) = \sum_{j=1}^{k} \sum_{x \in X_j} \ham(x, c_j'')\\
&=& \sum_{j \in I_s} \sum_{x \in X_j} \ham(x, c_j'') + \sum_{j \in I_l} \sum_{x \in X_j} \ham(x, c_j'') \\
&\leq& \sum_{j \in I_s} \sum_{x \in X_j} \left( \ham(x, n(x)) + \ham(n(x), c_j'') \right) + \sum_{j \in I_l} \sum_{x \in X_j} \ham(x, c_j'') \qquad \textrm{(using triangle inequality)}\\
&\leq& \sum_{j \in I_s} \Phi(B, X_j) + \sum_{j \in I_s} \sum_{x \in X_j'} \ham(x, c_j'') + \sum_{j \in I_l} \sum_{x \in X_j} \ham(x, c_j'') \\
&=& \sum_{j \in I_s}\Phi(B, X_j) + \sum_{j \in I_s} \sum_{x \in X_j'} \ham(x, c_j'')  + \sum_{j \in I_l} \left( \sum_{x \in X_j^{near}} \ham(x, c_j'') + \sum_{x \in X_j^{far}} \ham(x, c_j'')\right) \\
&\leq& \sum_{j \in I_s} \Phi(B, X_j) + \sum_{j \in I_s} \sum_{x \in X_j'} \ham(x, c_j'')  + \sum_{j \in I_l} \left( \sum_{x \in X_j^{near}} \left(\ham(x, n(x)) + \ham(n(x), c_j'') \right) + \sum_{x \in X_j^{far}} \ham(x, c_j'')\right) \\
&& \textrm{(using triangle inequality)}\\
&\leq& \sum_{j \in I_s} \Phi(B, X_j) + \sum_{j \in I_s} \sum_{x \in X_j'} \ham(x, c_j'')  + \sum_{j \in I_l} \left( \bar{n}_j R_j + \sum_{x \in X_j'} \ham(x, c_j'')\right) \\
&\leq& \sum_{j \in I_s} \Phi(B, X_j) + \sum_{j \in I_l} \bar{n}_j R_j + \sum_{j=1}^{k} \Phi(c_j'', X_j')\\
&\leq& \sum_{j \in I_s} \Phi(B, X_j) + \sum_{j \in I_l} \bar{n}_j R_j + \left( 1 + \frac{\veps}{16}\right) \cdot \sum_{j=1}^{k} \Phi(c_j', X_j')\\
%
%
&\leq& \sum_{j \in I_s} \Phi(B, X_j) + \sum_{j \in I_l} \bar{n}_j R_j + \left( 1 + \frac{\veps}{16}\right) \cdot \left( OPT + \sum_{j \in I_s} \Phi(B, X_j) + \sum_{j \in I_l} \bar{n}_j R_j \right)\\
&& \textrm{(using Lemma~\ref{lemma:combined})}\\
&=& \left(2 + \frac{\veps}{16}\right) \cdot \sum_{j \in I_s} \Phi(B, X_j) + \left(2 + \frac{\veps}{16}\right) \cdot \sum_{j \in I_l} \bar{n}_j R_j +  \left(1 + \frac{\veps}{16}\right) \cdot OPT\\
&\leq& \left(2 + \frac{\veps}{16}\right) \cdot \frac{\veps}{6} \cdot OPT + \left(2 + \frac{\veps}{16}\right) \cdot \sum_{j \in I_l} \bar{n}_j R_j +  \left(1 + \frac{\veps}{16}\right) \cdot OPT \qquad \textrm{(using the defn. of $I_s$)}\\
&\leq& \left(2 + \frac{\veps}{16}\right) \cdot \frac{\veps}{6} \cdot OPT + \left(2 + \frac{\veps}{16}\right) \cdot \sum_{j \in I_l} \frac{\veps}{4} \cdot \Delta_j +  \left(1 + \frac{\veps}{16}\right) \cdot OPT \qquad \textrm{(using Lemma~\ref{lemma:inter2})}\\
&\leq& (1 + \veps) \cdot OPT
\end{eqnarray*}
This completes the proof of the lemma.\qed
\end{proof}

So, now let us focus on obtaining a center set $C'' = \{c_1'', ...., c_k''\}$ that satisfies $\mathcal{R}$ and that satisfies $\sum_{j=1}^{k} \Phi(c_j'', X_j') \leq (1+\frac{\veps}{16}) \cdot \sum_{j=1}^{k} \Phi(c_j', X_j')$.
The sampling result (Theorem~\ref{thm:fomin}) of Fomin \etal~\cite{fomin18} tells us that for this bound to hold, it is sufficient to obtain subsets $T_1, ..., T_k$ such that for every $j$, $|T_j| = \tau$ and $T_j$ is a uniformly sampled set from the multi-set $X_j'$. 
We will argue that at least one of the sets $S_1, ..., S_k$ considered by the algorithm in line (5) satisfies this condition. 
For indices in the set $I_s$, this is easy to argue. 
This is because for any $j \in I_s$, $X_j'$ consists of only points from the set $B$ and since we take $\tau$ copies of every element of $B$, one of the sets $S_j$ considered in line (5) will be a uniform sample from $X_j'$.
However, for indices in $I_l$ it becomes a bit tricky. 
For any $j \in I_l$, $X_j'$ consists of points from the set $B$ and $X_j^{far}$. 
We will argue that every point in $X_j^{far}$ has some minimum probability of being sampled. Then we will argue that due to adequate oversampling and adequate number of copies of $B$ in the set $M_j = B_j' \cup M_j'|_{X_j}$, one of the subsets $S_j$ in line (5) will be a uniform sample from the set $X_j'$.

Toward this, the first observation we make is that the probability of sampling an element from $X_j^{far}$ for $j \in I_l$ is reasonably large (proportional to $\frac{\veps}{k}$). 
Using this fact, we show how to sample from $X_j'$ (almost uniformly). 
Finally, we show how to convert this almost uniform sampling to uniform sampling (at the cost of increasing the size of sample).

\begin{lemma}
\label{lem:osample}
Let $j \in I_l$.
Let $x$ be a sample from $D^2$-sampling w.r.t. $B$.
Then, $\pr[x \in X_j^{far}] \geq \frac{\eps}{8 \alpha k}$.
Further, for any point $p \in X_j^{far}$, $\pr[x=p] \geq \frac{\gamma}{|X_j|}$, where $\gamma = \frac{\veps^2}{54 \alpha k}$.
\end{lemma}

\begin{proof}
Note that $\sum_{p \in X_j^{near}} \pr[x=p] \leq \frac{R_j}{\Phi(B, X)} \cdot |X_j| \leq \frac{\eps}{9} \cdot \frac{\Phi(B, X_j)} {\Phi(B, X)}$.
Using the fact that $j \in I_l$, we have:
$$\pr[x \in X_j^{far}] \geq \pr[x \in X_j] - \pr[x \in X_j^{near}] \geq \frac{\Phi(B, X_j)}{\Phi(B, X)} - \frac{\eps}{9} \cdot \frac{\Phi(B, X_j)} {\Phi(B, X)} \geq \frac{\eps}{8 \alpha k}.$$

\noindent
Also, if $x \in X_j^{far}$, then $\Phi(B, \{x\}) \geq R_j=\frac{\eps}{9} \cdot \frac{\Phi(B, X_j)}{|X_j|}$.
Therefore,
$$\frac{\Phi(B, \{x\})}{\Phi(B, X)}  \geq \frac{\eps}{6 \alpha k} \cdot \frac{R_j}{\Phi(B, X_j)} \geq \frac{\veps}{6 \alpha k} \cdot \frac{\veps}{9} \cdot \frac{1}{|X_j|} \geq \frac{\veps^2}{54 \alpha k} \cdot \frac{1}{|X_j|}.
$$
This completes the proof of the lemma.
\qed
\end{proof}

Let $O_1, \ldots O_{\eta}$ be $\eta$ points sampled independently using $D^2$-sampling w.r.t. $B$.
We construct a new set of random variables $Y_1, \ldots, Y_{\eta}$.
Each variable $Y_u$ will depend on $O_u$ only, and will take values either in $X_j'$ or will be $\nl$.
These variables are defined as follows: if $O_u \notin X_j^{far}$, we set $Y_u$ to  $\nl$. 
Otherwise, we assign $Y_u$ to one of the following random variables with equal probability:
(i) $O_u$ or (ii) a random element of the multi-set $n(X_j^{near})$ defined as $n(X_j^{near}) \equiv \{b \in B : x \in X_j^{near} \textrm{ and } n(x) = b\}$.
The following observation follows from Lemma~\ref{lem:osample}.

\begin{corollary}
\label{cor:osample}
Let $j \in I_l$. For a fixed index $u$, and an element $x \in X_j'$, $\pr[Y_u=x] \geq \frac{\gamma'}{|X_j'|},$ where $\gamma'=\gamma/2$.
\end{corollary}

\begin{proof}
If $x \in X_j^{far}$, then we know from Lemma~\ref{lem:osample} that $O_u$ is $x$ with probability at least $\frac{\gamma}{|X_j'|}$ (note that
$X_j'$ and $X_j$ have the same cardinality). 
Conditioned on this event, $Y_u$ will be equal to $O_u$ with probability $1/2$.
Now suppose $x \in n(X_j^{near})$. Lemma~\ref{lem:osample} implies that $O_u$ is an element of $X_j^{far}$ with probability at least $\frac{\eps}{8 \alpha k}$.
Conditioned on this event, $Y_u$ will be equal to $x$ with probability at least $\frac{1}{2} \cdot \frac{1}{|n(X_j^{near})|}$. 
Therefore, the probability that $O_u$ is equal to $x$ is at least $\frac{\eps}{8 \alpha k} \cdot \frac{1}{2|n(X_j^{near})|} \geq \frac{\eps}{16 \alpha k |X_j'|} \geq \frac{\gamma'}{|X_j'|}$.
\qed
\end{proof}

Corollary~\ref{cor:osample} shows that we can obtain samples from $X_j'$ which are nearly uniform (up to a constant factor).
To convert this to a set of uniform samples, we use the idea of~\cite{jks}.
For an element $x \in X_j'$, let $\gamma_x$ be such that $\frac{\gamma_x}{|X_j'|}$ denotes the probability that the random variable $Y_u$ is equal to $x$ (note that this is independent of $u$).
Corollary~\ref{cor:osample} implies that $\gamma_x \geq \gamma'$.
We define a new set of independent random variables $Z_1, \ldots, Z_{\eta}$.
The random variable $Z_u$ will depend on $Y_u$ only.
If $Y_u$ is $\nl$, $Z_u$ is also $\nl$.
If $Y_u$ is equal to $x \in X_j'$, then $Z_u$ takes the value $x$ with probability $\frac{\gamma'}{\gamma_x}$, and $\nl$ with the remaining probability.
\lv{Note that $Z_u$ is either $\nl$ or one of the elements of $\ostp{\ti}$.
Further, conditioned on the latter event, it is a uniform sample from $\ostp{\ti}$.}
We can now prove the key lemma.

\begin{lemma}\label{lem:final}
Let $j \in I_l$. Let $\eta$ be $\frac{2 \tau}{\gamma'}$, and $\mathcal{Z}$ denote the multi set consisting of the non-null samples from $Z_1, \ldots, Z_{\eta}$. Then, with probability at least $(1-\frac{1}{k})$, the following holds (i) $|\mathcal{Z}| \geq \tau$ and (ii) $\mathcal{Z}$ is an iid sample from $X_j'$.
\end{lemma}

\begin{proof}
Note that a random variable $Z_u$ is equal to a specific element of $X_j'$ with probability equal to $\frac{\gamma'}{|X_j'|}$.
Therefore, it takes $\nl$ value with probability $1-\gamma'$.
Now consider a different set of iid random variables $Z_u'$, $1 \leq u \leq \eta$ as follows: each $Z_u$ tosses a coin with probability of Heads being $\gamma'$.
If we get Tails, it gets value $\nl$, otherwise it is equal to a random element of $X_j'$. 
It is easy to check that the joint distribution of the random variables $Z_u'$ is identical to that
of the random variables $Z_u$.
Thus, it suffices to prove the statement of the lemma for the random variables $Z_u'$.

Note that the non-null samples of $Z_1', ..., Z_{\eta}'$ are iid samples from $X_j'$. So, property (ii) of the lemma is satisfied by $\mathcal{Z}$. Now let us focus on property (i).
For this, we condition on the coin tosses of the random variables $Z_u'$.
Let $n'$ be the number of random variables which are not $\nl$.
($n'$ is a deterministic quantity because we have conditioned on the coin tosses).
Observe that the expected number of non-$\nl$ random variables is $\E[n'] = \gamma' \cdot \eta \geq 2 \tau$.
Therefore, with probability at least $(1 - \frac{1}{k})$ (using Chernoff-Hoeffding), the number of non-$\nl$ elements will be at least $\tau$. This completes the proof of the lemma.
\qed
\end{proof}

Let $B^{(\eta)}$ denotes the multi-set obtained by taking $\eta$ copies of each of the centers in $B$. 
Now observe that all the non-$\nl$ elements among $Y_1, \ldots, Y_{\eta}$ are elements of $\{O_1, \ldots, O_{\eta}\} \cup B^{(\eta)}$, and so the same must hold for $Z_1, \ldots, Z_{\eta}$. 
Moreover, since we only need a uniform subset of size $\tau$, $B_j'$ suffices instead of $B^{(\eta)}$.
This implies that in steps (5) of the algorithm ${\tt GoodCenters}$, we would have tried a set $S_j$ of size $\tau$ that is an iid sampled set from $X_j'$.

Combining the argument for index sets $I_s$ and $I_j$, we get that with probability at least $(1 - \frac{1}{k})^k$, there exists subsets $S_1, ..., S_k$ that the algorithm ${\tt GoodCenters}$ constructs in line (5) such that for every $i$, $S_i$ is an iid sample (with replacement) from $X_j's$. 
The sampling result of Fomin \etal and Ban \etal, that is Theorem~\ref{thm:fomin}, tells us that this is sufficient to show that the center set $C'' = \{c_1'', ..., c_k''\}$ so obtained in line (7) of ${\tt GoodCenters}$ satisfies eqn. (\ref{eqn:inter1}) with probability at least $\frac{\veps}{32}$.
Furthermore, this implies from Lemma~\ref{lemma:main-lem}, that $\Phi(C'', X) \leq (1+\veps) \cdot OPT$.
In summary, we get that in one iteration of the algorithm the probability of obtaining a $(1+\veps)$-approximate solution is at least $\left( 1 - \frac{1}{k}\right)^k \cdot \frac{\veps}{32}$. 
So a repeating $O(\frac{1}{\veps})$ times (as in line (2)) will give us a $(1 + \veps)$-approximate solution with high probability (probability at least $3/4$). This completes the proof of our main theorem~\ref{thm:main-theorem}.

\section{Conclusion and open problems}
In this work, we gave a 4-pass, polylog-space streaming PTAS for generalised binary $\ell_0$-rank-$r$ approximation problem and a 3-pass, polylog-space streaming PTAS for the constrained binary $k$-means problem (for $r$ and $k$ being constants independent of the problem size).
We do this by designing a PTAS for the latter problem and use a reduction from the former to the latter to obtain a PTAS also for the binary $\ell_0$-low rank approximation problem.
For the constrained binary $k$-means problem, we observe that as long as it is possible to obtain uniform samples from the target optimal clusters simultaneously, one can use previous sampling results of Ban \etal~\cite{ban} and Fomin \etal~\cite{fomin18} to obtain a PTAS. 
The $D^2$-sampling technique developed in the context of the standard constrained $k$-means problem~\cite{dx15,bjk,gjk} allows us to obtain something close to the above objective (of being able to sample uniformly from all target clusters) using a constant factor approximate solution for the binary $k$-means problem. This turns out to be sufficient to obtain a PTAS.

There are important questions that remain open.
The first question is whether one can design a PTAS with fewer passes? 
Furthermore, our constant pass algorithm requires polylogarithmic space. 
Note that this is mainly due to the fact that we need to guess the approximate size of each of the $k$ clusters. 
In the batch setting, Fomin~\etal~\cite{fomin18} have managed to get around the problem of guessing the size of each of the clusters by using an iterative {\it peeling} approach of Kumar \etal~\cite{kss}. 
However, their algorithm cannot be converted to a streaming PTAS that work in few passes.
So, an interesting question is whether a logspace algorithm is possible. 
In summary, the following open question would be nice to resolve:
\begin{quote}
{\it Does there exist a single-pass, logspace PTAS for the constrained binary $k$-means problem and the generalised binary $\ell_0$-rank-$r$ approximation problem?}
\end{quote}


\addcontentsline{toc}{section}{References}
\bibliographystyle{alpha}
\bibliography{paper}

\end{document}